\documentclass[runningheads]{11llncst}
\usepackage[T1]{fontenc}
\usepackage{graphicx}
\usepackage{amsmath}
\usepackage{amsfonts}
\usepackage{bm}
\usepackage{booktabs}
\usepackage{tabularx}
\usepackage{adjustbox}
\usepackage{marvosym}
\begin{document}

\title{Facility Location Games for Multi-Location Agents with Satisfaction}

\author{Huanjun Wang\inst{1} \and
Qizhi Fang\inst{1} \and
Wenjing Liu\textsuperscript{\Letter}\inst{1,2}}
\authorrunning{H. Wang et al.}
%
\institute{School of Mathematical Sciences, Ocean University of China, Qingdao, 266100, China \\
\email{hjwang@stu.ouc.edu.cn, \{qfang, liuwj\}@ouc.edu.cn}\\
\and
Laboratory of Marine Mathematics, Ocean University of China, Qingdao, 266100, China\\
\email{liuwj@ouc.edu.cn}\\}

\maketitle              

\begin{abstract}
     In this paper, we study mechanism design for single-facility location games where each agent has multiple private locations in $[0,1]$. The individual objective is a satisfaction function that measures the discrepancy between the optimal facility location for an agent and the location provided by the mechanism. Based on  different distance functions from  agents to the facility, we consider two types of individual objectives: the sum-variant satisfaction and the max-variant satisfaction. Our goal is to design mechanisms that locate one facility to maximize  the sum (or the minimum) of all agents’ satisfactions, while  incentivizing agents to truthfully report their locations.    In this paper, we mainly focus on desirable and obnoxious facility location games.   For desirable facility location games, we propose two group strategy-proof mechanisms with approximation ratios of $2$ and $\frac{5}{4}$ for maximizing the sum of the sum-variant and  max-variant satisfaction, respectively. Moreover, another mechanism achieves an approximation ratio of 2 for simultaneously maximizing the minimum of the sum-variant satisfaction  and the minimum of the  max-variant satisfaction.   For  obnoxious facility location games, we establish that two group strategy-proof mechanisms are the best possible,  providing an approximation ratio of $2$ for maximizing the sum of the sum-variant satisfaction and the sum of the  max-variant satisfaction, respectively. Additionally, we devise two $\frac{4}{3}$-approximation randomized group strategy-proof mechanisms, and provide two lower bounds of 1.0625 and 1.0448 of randomized strategy-proof mechanisms for maximizing the sum of the sum-variant satisfaction and the sum of the max-variant satisfaction, respectively.
    
    \keywords{Facility location games \and Mechanism design\and Satisfaction \and Strategy-proof \and Approximation ratio.}
\end{abstract}



\section{Introduction}
The facility location problem, as one of the fundamental optimization problems, mainly focuses on determining the optimal location for the facility to optimize some social objective under certain constraints. This problem usually involves the location selection of public service facilities, logistics centers, factories or commercial facilities, etc. In many real-world settings, agents' locations may be their private information. The social planner must determine the facility location based on the locations reported by agents. Agents may misreport their location information to optimize their individual objectives. Therefore, the goal of the social planner is to establish a mechanism to determine the location of the facility while approximately optimizing some social objective and enabling agents to report truthfully. 

In this paper, we study facility location games for multi-location agents with satisfaction. In most previous studies on facility location games, the cost or utility of each agent is defined by the distance to the facility location. However, agents at distinct spatial positions exhibit divergent levels of satisfaction with respect to identical facility distances, owing to variations in their achievable maximum utilities (or minimal costs). This discrepancy stems from perceptional differences influenced by spatial context. For instance, urban core residents and suburban inhabitants perceive the same physical distance differently. A given distance from the city center to the suburbs is often regarded as considerable by urban dwellers, as most of their daily necessities are accessible within a compact central area. In contrast, suburban residents, who are accustomed to longer travel distances for routine activities such as commuting and shopping, perceive the same distance as relatively acceptable.

\vspace{-0.3cm}

\subsection{Related Work} 
Mechanism design for facility location games has been extensively studied over the past decades.
Procaccia and Tennenholtz \cite{procaccia2009approximate} first studied approximate mechanism design without money for facility location games. They considered strategy-proof mechanisms on the line for $1$-facility and $2$-facility for minimizing the social cost and the maximum cost, and proposed a natural extension where each agent controls multiple locations (referred to as multi-location agents). The agent's cost  is either the total distance or the maximum distance from her locations to the facility. 

Followed from  \cite{procaccia2009approximate}, mechanism design has been widely discussed in different facility location game models. We first review  studies of single-location agents (i.e., agents controlling exactly one location), then focus on two directions closely related to this work: multi-location agents and satisfaction.

\subsubsection{Facility Location Games with Single-Location Agents.}  Lu et al.  \cite{lu2010asymptotically} extended the line setting in \cite{procaccia2009approximate} to the circle for two-facility location problem under minimizing the social cost.  Alon et al.  \cite{alon2010strategyproof} studied the  facility location problem on a circle and general metrics for minimizing the maximum cost. The aforementioned studies all concern desirable  facility location games, that is, each agent prefers to be close to facilities. Ibara and Nagamochi \cite{ibara2012characterizing} characterized  strategy-proof mechanisms for  obnoxious facility location games. Cheng et al. \cite{cheng2013strategy} provided majority mechanisms for obnoxious facility location games  on the line, tree, circle, and general metric space.  Additionally, many other  facility location game models have been proposed under different assumptions. Feigenbaum and Sethuraman  \cite{feigenbaum2014strategyproof} discussed a preference model  where some agents want to be close to the facility while others  want to stay away. Zou and Li \cite{zou2015facility} studied two-opposite-facility location games with limited distance. Fong et al.  \cite{fong2018facility} proposed a facility location model where each agent has fractional preference to facilities: the agents' usage frequencies of two facilities are characterized by two fractions that sum to one. In the above research results, the facility can be built in continuous space. But in some cases, facilities can only be built at some limited locations. Feldman et al. \cite{feldman2016voting} first studied the single-facility location problem  in the context of voting for minimizing the social cost. Tang et al. \cite{tang2020mechanism} further considered  two-facility location games for minimizing the maximum cost.

\subsubsection{Facility Location Games with Multi-Location Agents.}  Thang \cite{thang2010group} extended the desirable facility location problem in \cite{procaccia2009approximate} to the metric space. Mei et al.  \cite{mei2018approximation} designed mechanisms for obnoxious facility location games in which each agent has multiple locations on a line, examining two individual objective functions: maximizing the sum of squared distances and maximizing the sum of distances. Feigenbaum and Sethuraman  \cite{feigenbaum2014strategyproof} investigated one-dimensional hybrid facility location games, in which some agents want to be close to facility and  others want to be far from the facility.

\subsubsection{Facility Location Games with Satisfaction.} Mei et al. \cite{mei2019facility} first considered  desirable facility  and obnoxious facility location games with happiness factor (commonly referred to as satisfaction in subsequent research) which measures the discrepancy between the optimal facility location
for an agent and the location provided by the mechanism.  They devised a $\frac{5}{4}$-approximation group strategy-proof mechanism for desirable facility location games, a $2$-approximation deterministic strategy-proof mechanism and a $\frac{4}{3}$-approximation randomized strategy-proof mechanism for obnoxious facility location games on a line for maximizing the social happiness factor. Walsh \cite{walsh2021strategy} extended the median mechanism to a facility location games at limited locations in which facilities can only be placed in finite intervals. Anastasiadis and Deligkas  \cite{anastasiadis2018heterogeneous} discussed heterogeneous facility location games with satisfaction where each agent wants to be  close to, stay away from, or be indifferent to facilities. The satisfaction in these models above is derived from the idea presented in \cite{mei2019facility}. Recently, Li et al. \cite{li2025strategy} analyzed facility location games with another way of quantifying the agent's satisfaction. According to the psychological theory of relative effects, each agent's satisfaction depends  not only on her distance to the facility but also on the probability that its location exceeds a "fair baseline".

\vspace{-0.3cm}
\subsection{Our Results}  
In this paper, we study mechanism design  for facility location games where each agent has multiple private locations in $[0,1]$. Building upon the satisfaction defined for single-location agents in \cite{mei2019facility} and the distance functions for multi-location agents in \cite{procaccia2009approximate} (namely, the sum of distances and the maximum distance between agents' locations and the facility), we introduce two corresponding types of  satisfaction  for multi-location agents: sum-variant satisfaction and max-variant satisfaction. Each agent's objective is to maximize her own satisfaction. We consider two social objectives: maximizing the social satisfaction (i.e., the sum of all agents' satisfactions), and maximizing the minimum satisfaction (i.e., the minimum  of all agents' satisfactions). 

In Section 3, we study desirable facility location games where each agent wants to be close to the facility. In Section 4, we study obnoxious facility location games  where each agent wants to be far away from the facility. The results  are summarized in the  Table 1 and Table 2,  where UB and LB indicate the upper bound and the lower bound on the approximation ratio for any strategy-proof mechanism, respectively.

\begin{table}
\caption{A summary of the results of  desirable facility location games.}\label{tab2}
\centering
\begin{tabularx}{\textwidth}{|X|X|X|}
\hline
  \textbf{Individual objectives} &\textbf{Social objectives} &\textbf{Deterministic}\\
\hline
 Sum-variant satisfaction  & Social satisfaction &   UB: 2 (Theorem 1) \\
     &  & LB: 1.086 \cite{mei2019facility}\\
    &  Minimum satisfaction  & UB: 2 (Theorem 2) \\
 & &  LB: $\frac{4}{3}$ \cite{mei2019facility}\\
  Max-variant satisfaction & Social satisfaction  & UB: $\frac{5}{4}$ (Theorem 3)  \\
 & & LB: 1.086 \cite{mei2019facility}\\  

 &  Minimum satisfaction  & UB: 2 (Theorem 4) \\
 &  & LB: $\frac{4}{3}$ \cite{mei2019facility}\\
\hline
\end{tabularx}
\end{table}

\vspace{-0.8cm}

\begin{table}
\caption{A summary of the  results of obnoxious facility location games.}
\centering

\begin{tabularx}{\textwidth}{|l|X|l|l|}
\hline%
 \textbf{Individual } &\textbf{Social objectives} & \textbf{Deterministic} &\textbf{Randomized} \\
  \textbf{ objectives} & &  & \\
\hline

 Sum-variant  & Social satisfaction & UB: 2 (Theorem 5) & UB: $\frac{4}{3}$  (Theorem 6) \\
    satisfaction &   & LB: 2 \cite{mei2019facility}  & LB: 1.0625 (Theorem 7)  \\
    Max-variant  & Social satisfaction & UB: 2 (Theorem 8) & UB:  $\frac{4}{3}$ (Theorem 9)\\
    satisfaction   & & LB: 2 \cite{mei2019facility} & LB: $1.0448$ (Theorem 10)\\
\hline
\end{tabularx}

\end{table}

Compared to the  desirable facility location games for single-location agents \cite{mei2019facility} (which achieves an approximation ratio of $\frac{3}{2}$), the median mechanism can only achieve an approximation ratio of 2 in the multi-location agents setting. For obnoxious facility games, the approximation ratios of deterministic and randomized mechanisms remain unchanged regardless of the number of locations controlled by the agents.

\section{Preliminaries}

Let $N=\{ 1,2,\dots,n\}$ be a set of agents located on the interval $I=[0,1]$. Each agent $i\in N$ has a private location profile $\textbf{x}_i=( x_{i1},\dots,x_{i\omega_i})$, where $\omega_i$ is the number of locations controlled by agent $i$. Without loss of generality, assume that $ x_{i1}\leq \dots\leq x_{i\omega_i}$ and $\omega=\omega_1+\cdots+\omega_n$. Let $\textbf{x}=(\textbf{x}_1,\dots,\textbf{x}_n )$ be the collected location profile. The social planner determines the location of a facility based on the locations reported by the agents. 

\subsubsection{Individual Objectives.} For the distance function of multi-location agents,  two types are provided  in \cite{procaccia2009approximate}: (1) $d_1(z,\textbf{x}_i)=\sum_{j=1}^{\omega_i}{|z-x_{ij}|}$, for $i\in N$, where agents are concerned about the total distance from the facility location to their locations; (2) $d_2(z,\textbf{x}_i)= \max_{1\leq j\leq \omega_i} {|z-x_{ij}|}$, for $i \in N$, where agents are concerned about the maximum distance from the facility location to their locations. Correspondingly, we propose two types of individual objective functions: the sum-variant satisfaction (with $d_1$) and the max-variant satisfaction (with $d_2$). 

In this paper, we mainly focus on desirable and obnoxious facility location games. For desirable facility location games, each agent wants to stay as close to the facility as possible. The agent will be completely satisfied if the facility is located at a location with  the minimum distance to her locations, and she will be least  satisfied if  the facility is located at a location with  the maximum distance to her locations. 

The satisfaction  $s(y,\textbf{x}_i)$ of agent $i$ is defined below, $$s(y,\textbf{x}_i)=1-\frac{d(y,\textbf{x}_i)-\delta(\textbf{x}_i)}{\Delta(\textbf{x}_i)-\delta(\textbf{x}_i)},$$
where $\Delta(\textbf{x}_i)=\max_{z\in I}d(z,\textbf{x}_i)$ and  $\delta(\textbf{x}_i)=\min_{z\in I}d(z,\textbf{x}_i)$ are the maximum and minimum possible distance from agent $i$ to the facility, respectively, and $d(z,\textbf{x}_i)$ denotes a distance function.

Analogously, for  obnoxious facility location games where each agent wants to stay as far from the facility as possible, the satisfaction of agent $i$ is given as follows,  $$s(y,\textbf{x}_i)=\frac{d(y,\textbf{x}_i)-\delta(\textbf{x}_i)}{\Delta(\textbf{x}_i)-\delta(\textbf{x}_i)}.$$

For each agent, the maximum possible distance will be obtained from the two endpoints 0 or 1; the minimum possible distance is obtained from her median location and the midpoint of her locations for $d_1$ and $d_2$, respectively.

Each agent wants to maximize her satisfaction. Given a facility location $y$, an  agent $i\in N$ and her location profile $\textbf{x}_i$, let $c_i=\frac{x_{i1}+x_{i\omega_i}}{2}$ be the midpoint of her locations. A crucial observation for the distance function $d_2$ is that  $$d_2(y,\textbf{x}_i)=\max_{1\leq j\leq \omega_i}|y-x_{ij}|=|y-c_i|+\frac{x_{i\omega_i}-x_{i1}}{2}.$$

\begin{proposition}
    For the  max-variant satisfaction, the satisfaction of each agent  is only determined by  the midpoint of her locations.
\end{proposition}

Due to space constraints, the proof of Proposition 1 can be found in the appendix.\footnote{Due to space constraints, all missing proofs can be found in the appendix.}

\subsubsection{Social Objectives.} We  focus on optimizing the following two social objective functions: (1) maximizing the social satisfaction, which is defined as the  sum of all the agents' satisfactions, that is, $SS(y,\textbf{x}) =\sum_{i\in N}s(y, \textbf{x}_i)$; (2) maximizing the minimum satisfaction, which is defined as the minimum of all the agents' satisfactions, that is, $MS(y,\textbf{x}) =\min_{i\in N}s(y, \textbf{x}_i)$. 

A deterministic mechanism is a function $f : I^\omega \to I$ that maps a given location profile $\textbf{x}\in I^\omega$ to a facility location $f(\textbf{x})\in I$. A randomized mechanism is a function $f$ that maps a given location profile $\textbf{x}\in I^\omega$  to a probability distribution  over $I$. The satisfaction  of agent $i$ for randomized mechanism $f$ is defined as  $s(f(\textbf{x}), \textbf{x}_i) = \mathbb{E}_{y\sim f(\textbf{x})}[s(y,\textbf{x}_i)]$. Moreover, we have that $SS(f(\textbf{x}),\textbf{x})=\mathbb{E}_{y\sim f(\textbf{x})}[SS(y,\textbf{x})] $ and $MS(f(\textbf{x}),\textbf{x})=\min_i\mathbb{E}_{y\sim f(\textbf{x})}[s(y,\textbf{x})] $.

A mechanism $f$ is strategy-proof, i.e., no agent can improve  her satisfaction by misreporting her locations, regardless of the others' strategies. Formally, for any $i\in N$, any location profile $\textbf{x}$ and any $\textbf{x}_i' \in I^{\omega_i}$, it holds  that $$s(f(\textbf{x}), \textbf{x}_i) \geq s(f(\textbf{x}'_i, \textbf{x}_{-i}), \textbf{x}_i),$$ where $\textbf{x}_{-i}= (\textbf{x}_1,\dots, \textbf{x}_{i-1}, \textbf{x}_{i+1},\dots, \textbf{x}_n)$ is the location profile without agent $i$. Moreover, a mechanism $f$ is group strategy-proof if no coalition of agents can improve their satisfactions by misreporting their locations simultaneously. Formally, for any non-empty subset $G \subseteq N$ of agents, any location profile $\textbf{x}$,  and  any $\textbf{x}'
_G$, there exists $i \in G$ such that  $$s(f(\textbf{x}), \textbf{x}_i) \geq s(f(\textbf{x}'_G, \textbf{x}_{-G}), \textbf{x}_i),$$  where $\textbf{x}_{-G}$ is the location profile without agents in $G$.

 A mechanism $f$ is $\alpha$-approximation if it satisfies
\begin{equation*}
    \alpha=\sup_{\textbf{x}}\frac{S(opt(\textbf{x}),\textbf{x})}{S(f(\textbf{x}), \textbf{x})},
\end{equation*}
where $S$ is either $SS$ or $MS$, and $opt(\textbf{x})$ is the corresponding optimal facility location. 

In this paper, we are interested in deterministic and randomized (group) strategy-proof mechanisms with small approximation ratios. Firstly, we study desirable facility location games where each agent wants  to be close to the facility. Then, we discuss obnoxious facility location games where each agent wants  to be far away from the facility.

\section{Desirable Facility Location Games}
In this section, we focus on deterministic (group) strategy-proof mechanisms for  desirable facility location games. Next, we will discuss the sum-variant  and  max-variant satisfaction, respectively.

\subsection{Sum-Variant Satisfaction}

In this subsection, we discuss the case where the individual objective is  the sum-variant satisfaction. We  propose a deterministic group strategy-proof mechanism  for maximizing the social satisfaction and  another  for maximizing the minimum satisfaction. 

Recall that  $d_1(y,\textbf{x}_i)=\sum_{j=1}^{\omega_i}|y-x_{ij}|$ for $i\in N$ in the sum-variant satisfaction. For each agent, the optimal facility location is  her median location. Inspired by this,  we propose a mechanism which outputs the median of all agents' median locations. 

Let $med(\textbf{x}_i)$ denote the median location of  locations controlled by agent $i$. If there are two median positions, we select the left median, that is, $med(\textbf{x}_i)=x_{i\lceil\frac{\omega_i}{2}\rceil}$. Denote the median location profile as $\textbf{x}'=( med(\textbf{x}_1), \dots,med(\textbf{x}_n))$. Without loss of generality, assume that $med(\textbf{x}_1)\leq \dots\leq med(\textbf{x}_n)$.

\begin{mechanism}\label{me1}
    Given a location profile $\textbf{x}$,  locate the facility at $med(\textbf{x}')$.
\end{mechanism}

\begin{theorem}
     Mechanism $\ref{me1}$ is  a group strategy-proof mechanism with approximation ratio of $2$ for maximizing the social satisfaction.
\end{theorem}

Mei et al. \cite{mei2019facility} have provided a lower bound of 1.086 for single-location agents under the social satisfaction objective. It implies that any deterministic strategy-proof mechanism for multi-location agents has an approximation ratio of at least 1.086 for maximizing the social satisfaction. 

For maximizing the minimum  satisfaction, Mechanism \ref{me1} has unbound approximation ratio. In fact, consider a location profile $\textbf{x}$ with three agents, where $\textbf{x}_1=(0,\frac{1}{2})$, $\textbf{x}_2=(0,1)$, and $\textbf{x}_3=(\frac{1}{2},1)$. Mechanism 1 locates the facility at 0, which gives the minimum satisfaction of 0. However, it holds that $opt(\textbf{x})=\frac{1}{2}$ and $MS(opt(\textbf{x}),\textbf{x})=1$. Then,  an unbounded approximation ratio is derived. Therefore, we propose another mechanism for maximizing the minimum  satisfaction.  

\begin{mechanism} \label{me2}
    Given a location profile $\textbf{x}$, locate the facility at $\frac{1}{2}$.
\end{mechanism}

\begin{theorem} \label{th2}
     Mechanism $\ref{me2}$ is a group strategy-proof mechanism with  approximation ratio of $2$ for maximizing the minimum   satisfaction. 
\end{theorem}

For maximizing the minimum satisfaction, Mei et al. \cite{mei2019facility} have given a lower bound of $\frac{4}{3}$ on the approximation ratio of deterministic strategy-proof mechanisms.

\subsection{Max-Variant Satisfaction}
In this subsection, we study the case where the individual objective is the max-variant satisfaction. We present a deterministic group strategy-proof mechanism  for maximizing the social satisfaction and another for maximizing the minimum satisfaction. 

Note that the optimal facility location for each agent is the midpoint of her locations under the max-variant satisfaction.  For each agent $i$, let $c_i$ be the midpoint of her location.  Let $\textbf{c}=( c_1, c_2, \cdots, c_n)$ be the midpoint location profile of agents. Without loss of generality, we next assume that $c_1\leq c_2\leq \cdots\leq c_n$.

\begin{mechanism}\label{me4}
     Given a location profile $\textbf{x}$, if $c_{\lceil\frac{n}{2}\rceil}\leq\frac{1}{5}$, locate the facility at $\frac{1}{5}$; if $c_{\lceil\frac{n}{2}\rceil}\geq\frac{4}{5}$, locate the facility at $\frac{4}{5}$; otherwise, locate the facility at $c_{\lceil\frac{n}{2}\rceil}$.
\end{mechanism}

\begin{theorem}
    Mechanism $\ref{me4}$ is a group strategy-proof mechanism with approximation ratio of $\frac{5}{4}$ for maximizing the social satisfaction. 
\end{theorem}

Furthermore, building on Proposition 1 and the lower bound of  1.086 in \cite{mei2019facility} for single-location agents, it implies that any deterministic strategy-proof mechanism has an approximation ratio of at least 1.086 for maximizing the social satisfaction.

\begin{theorem}
    Mechanism $2$ is a group strategy-proof mechanism with approximation ratio of $2$ for maximizing the minimum satisfaction.
\end{theorem}

For the lower bound, we can obtain the same conclusion as the sum-variant satisfaction: any  deterministic strategy-proof mechanism has an approximation ratio of at least $\frac{4}{3}$ for maximizing the minimum satisfaction.

\section{Obnoxious Facility Location Games}
In this section, we discuss  deterministic  and randomized (group) strategy-proof mechanisms for obnoxious facility location games.  Recall that the individual objective of agent $i$ about the obnoxious facility location $y$ is $$s(y,\textbf{x}_i)=\frac{d(y,\textbf{x}_i)-\delta(\textbf{x}_i)}{\Delta(\textbf{x}_i)-\delta(\textbf{x}_i)}.$$

\subsection{Sum-Variant Satisfaction}

In this subsection, we consider the case where the individual objective is the sum-variant satisfaction. We first present a deterministic group strategy-proof mechanism for maximizing the social satisfaction. Additionally, we propose a randomized group strategy-proof mechanism and  provide the lower bound of $\frac{17}{16}$ on the approximation ratio of any randomized strategy-proof mechanism.

Given a location profile $\textbf{x}$, let  $N_1(\textbf{x})=\{i\in N|\sum_{j=1}^{\omega_i}x_{ij} \geq \sum_{j=1}^{\omega_i}(1-x_{ij})\}$ be the set of agents who prefer $0$, $N_2(\textbf{x)}=\{i\in N|\sum_{j=1}^{\omega_i}x_{ij} < \sum_{j=1}^{\omega_i}(1-x_{ij})\}$ be the set of agents who prefer $1$.  

\begin{mechanism}\label{th3}
    Given a location profile $\textbf{x}$, if $|N_1(\textbf{x})|\geq |N_2(\textbf{x})|$, locate the facility at $0$;  otherwise, locate the facility at $1$.
\end{mechanism}

\begin{theorem}
    Mechanism $\ref{th3}$ is a group strategy-proof mechanism with approximation ratio of $2$ for maximizing the social satisfaction. 
\end{theorem}

Mei et al. \cite{mei2019facility} proved that the approximation ratio of any deterministic strategy-proof mechanism is at least 2 for obnoxious facility location games with single-location agents. Thus, Mechanism \ref{th3} is the best deterministic strategy-proof mechanism. Now, we consider randomized mechanisms. Let $n_1$ and $n_2$ be the number of agents in sets $N_1(\textbf{x})$ and $N_2(\textbf{x})$, respectively.

\begin{mechanism} \label{th4}
    Given a location profile $\textbf{x}$, locate the facility at  $0$ with probability $\frac{n_1}{n}$ and $1$ with probability $\frac{n_2}{n}$. 
\end{mechanism}

\begin{theorem}\label{th7}
	Mechanism $\ref{th4}$ is a group strategy-proof mechanism with approximation ratio of $\frac{4}{3}$ for maximizing the social satisfaction. 
\end{theorem}

\begin{proof}
	Let $f$ denote Mechanism \ref{th4}. Given any location profile $\textbf{x}$,  we first consider the special case: $n_1=0$ or $n_2=0$. If $n_1=0$ or $n_2=0$, Mechanism \ref{th4} returns the facility location $y=1$ or $y=0$ with probability 1. No agent would like to misreport her locations. Furthermore, when $n_1=0$ or $n_2=0$, $SS(f(\textbf{x}),\textbf{x})=SS(opt(\textbf{x}),\textbf{x})$.
    
    Next, we assume that $n_1\neq0$ and $n_2\neq0$. Let $G\subseteq N$ denote a coalition. Suppose that agents in $G$ change the output of mechanism by misreporting their locations, namely  $n'_1>n_1$ or $n_1>n'_1$. $\textbf{x}'$ is the new location profile.
    
    If $n'_1>n_1$, then there at least exists an agent $i\in G$ in $N_2(\textbf{x})$ but she reports falsely in $N_1(\textbf{x})$. Then, by $s(1,\textbf{x}_i)=1$ and $n_1'-n_1=n_2-n_2'$, we have that 
    \begin{equation*}
      \begin{aligned}
          s(f(\textbf{x}),\textbf{x}_i)-s(f(\textbf{x}'),\textbf{x}_i)&=\frac{n_1}{n}s(0,\textbf{x}_i)+\frac{n_2}{n}s(1,\textbf{x}_i)-\frac{n_1'}{n}s(0,\textbf{x}_i)-\frac{n_2'}{n}s(1,\textbf{x}_i)\\
          &=\frac{n_1-n_1'}{n}s(0,\textbf{x}_i)+\frac{n_2-n_2'}{n}\\
          &=\frac{n_1-n_1'}{n}s(0,\textbf{x}_i)+\frac{n_1'-n_1}{n}\\
          &=\frac{n_1-n_1'}{n}(s(0,\textbf{x}_i)-1)>0. 
      \end{aligned}  
 \end{equation*}
This implies that agent $i$'s satisfaction strictly decreases. 

Similarly, if $n_1>n_1'$, then there at least exists an agent $j\in G$  in $N_1(\textbf{x})$ but she reports falsely in $N_2(\textbf{x})$. However, agent $j$ cannot increase her satisfaction.  Thus, Mechanism \ref{th4} is group strategy-proof. 
    
    We now show the approximation ratio of Mechanism \ref{th4}. The optimal facility location is either 0 or 1. Without loss of generality, assume that the optimal facility location is at 0. Let  $f(\textbf{x})$ be the probability distribution obtained by Mechanism \ref{th4}, it holds that 
	\begin{equation*}
		\begin{aligned}
			SS(f(\textbf{x}),\textbf{x})&=\mathbb{E}_{y\sim f(\textbf{x})}[SS(y,\textbf{x})]=\frac{n_1}{n} SS(0,\textbf{x})+\frac{n_2}{n} SS(1,\textbf{x})\\
            &\geq  \frac{n_1}{n} SS(opt(\textbf{x}),\textbf{x})+ \frac{n^2_2}{n^2}SS(opt(\textbf{x}),\textbf{x})\\
			&=\left[\left(\frac{n_2}{n}-\frac{1}{2}\right)^2+\frac{3}{4}\right]SS(opt(\textbf{x}),\textbf{x})\\
            &\geq\frac{3}{4}SS(opt(\textbf{x}),\textbf{x}).
		\end{aligned}
	\end{equation*}
     The first inequality holds since $SS(1,\textbf{x})\geq n_2$ and $SS(opt(\textbf{x}),\textbf{x})\leq n$. 
    
    Moreover, the approximation ratio of Mechanism \ref{th4} is tight. Consider a location profile $\textbf{x}$ with two agents, where $\textbf{x}_1=(0,1)$ and $\textbf{x}_2=(0,\frac{1}{2})$. For this profile, we can show that $opt(\textbf{x})=1$ and $SS(opt(\textbf{x}),\textbf{x})=2$. However, by Mechanism \ref{th4}, it holds that $SS(f(\textbf{x}),\textbf{x})=\frac{3}{2}$. Thus, the approximation ratio of Mechanism \ref{th4} is $\frac{4}{3}$.
    \qed
\end{proof}

\begin{theorem}
    Any randomized strategy-proof  mechanism has an approximation ratio of at least $\frac{17}{16}\approx1.0625$ for maximizing the social satisfaction.
\end{theorem}

\subsection{Max-Variant Satisfaction}

In this subsection, we explore the case where the individual objective is the max-variant satisfaction. We present a deterministic group strategy-proof mechanism which is the best for maximizing the social satisfaction. Furthermore, we propose a randomized group strategy-proof mechanism and derive a lower bound on the approximation ratio of any randomized strategy-proof mechanism. 

Given a location profile $\textbf{x}$, let $S_1(\textbf{x})$ be the set of agents with midpoint location in $[0,\frac{1}{2}]$, $S_2(\textbf{x})$ be the set of agents with midpoint location in $(\frac{1}{2},1]$.

\begin{mechanism}\label{me5}
    Given a location profile $\textbf{x}$, if $|S_1(\textbf{x})| \geq |S_2(\textbf{x})|$, locate the facility at 1; otherwise, locate the facility at $0$.
\end{mechanism}

\begin{theorem}\label{th8}
    Mechanism $\ref{me5}$ is a group strategy-proof mechanism with approximation ratio of $2$ for maximizing the social satisfaction. 
\end{theorem}

   Based on the results of Mei et al. \cite{mei2019facility} and Proposition 1, we conclude that Mechanism \ref{me5} is the best deterministic strategy-proof mechanism.  Now, we consider a randomized mechanism.

\begin{mechanism} \label{me6}
	 Given a location profile $\textbf{x}$, locate the facility at  $1$ with probability $\frac{|S_1(\textbf{x})|}{n}$ and $0$ with probability $\frac{|S_2(\textbf{x})|}{n}$. 
\end{mechanism}

\begin{theorem}\label{th9}
	Mechanism $\ref{me6}$ is   a group strategy-proof  mechanism with approximation ratio of $\frac{4}{3}$ for maximizing the social satisfaction. 
\end{theorem}
\begin{proof}
    The proof of Theorem \ref{th9} follows a similar approach to that of Theorem \ref{th7}. \qed
\end{proof}

As discussed before, when considering the max-variant satisfaction, only the midpoint location of each agent needs to be taken into account. Based on this observation, we obtain the lower bound  of any randomized strategy-proof mechanism under the social satisfaction objective.

\begin{theorem}\label{th5}
		Any  randomized strategy-proof mechanism has an approximation ratio of at least $\frac{4\sqrt{2}}{4+\sqrt{2}}\approx1.0448$ for maximizing the social satisfaction. 
\end{theorem}

\begin{proof}
    Consider a midpoint location profile $\textbf{c}=(\epsilon,1-\epsilon)$, where $\epsilon<\frac{1}{2}$, the optimal facility location is $opt(\textbf{c})=0$. Then, $SS(opt(\textbf{c}),\textbf{c})=\frac{\epsilon}{1-\epsilon}+\frac{1-\epsilon}{1-\epsilon}=\frac{1}{1-\epsilon}$. Assume that there exists a  randomized  strategy-proof mechanism  $f$ with approximation ratio $\alpha$. Note that $$SS(f(\textbf{c)},\textbf{c})=\mathbb{E}_{y\sim f(\textbf{c})}\left[\frac{|y-\epsilon|}{1-\epsilon}+\frac{|y-(1-\epsilon)|}{1-\epsilon}\right]\leq \frac{1}{1-\epsilon}=SS(opt(\textbf{c}),\textbf{c}).$$
		
		We can get $\mathbb{E}_{y\sim f(\textbf{c})}[|y-\epsilon|+|y-(1-\epsilon)|]\leq 1$. Without loss of generality, we assume that $\mathbb{E}_{y\sim f(\textbf{c})}[|y-\epsilon|]\leq \frac{1}{2}$.
		
		Consider a new location profile $\textbf{c}'=(0,1-\epsilon)$, which means that agent $1$ misreports her location to $0$ rather than her true location $\epsilon$. In $\textbf{c}'$, the optimal facility location is at $1$, yielding the social satisfaction $SS(opt(\textbf{c}'),\textbf{c}')=\frac{1}{1-\epsilon}$. Since the mechanism $f$ is strategy-proof, then 
		\begin{equation*}
			s(f(\textbf{c}'),\epsilon)=\mathbb{E}_{y\sim f(\textbf{c}')}\left[\frac{|y-\epsilon|}{1-\epsilon}\right]\leq \mathbb{E}_{y\sim f(\textbf{c})}\left[\frac{|y-\epsilon|}{1-\epsilon}\right]=s(f(\textbf{c}),\epsilon).
		\end{equation*}
		It means $\mathbb{E}_{y\sim f(\textbf{c}')}[|y-\epsilon|]\leq \mathbb{E}_{y\sim f(\textbf{c})}[|y-\epsilon|]\leq \frac{1}{2}$.
		
		Let $p(y\geq 1-\epsilon)=q$ in $f(\textbf{c}')$. It holds that 
		\begin{equation*}
			\begin{aligned}
				\frac{1}{2} &\geq \mathbb{E}_{y\sim f(\textbf{c}')}[|y-\epsilon|]=\mathbb{E}_{y\sim f(\textbf{c}')}[|y-\epsilon|| y<1-\epsilon](1-q)\\
				&+\mathbb{E}_{y\sim f(\textbf{c}')}[|y-\epsilon|| y\geq 1-\epsilon]q\geq (1-2 \epsilon)q,
			\end{aligned}
		\end{equation*}
	 therefore, $q\leq \frac{1}{2(1-2\epsilon)}$.

		Note that in $\textbf{c}'$, 
		\begin{equation*}
			\begin{aligned}
				\mathbb{E}_{y\sim f(\textbf{c}')}\left[SS(y,\textbf{c}')| y<1-\epsilon\right]&=\mathbb{E}_{y\sim f(\textbf{c}')}\left[y+\frac{|y-(1-\epsilon)|}{1-\epsilon}\bigg| y<1-\epsilon\right]\\
				&\leq \frac{1}{1-\epsilon}\mathbb{E}_{y\sim f(\textbf{c}')}\left[y+|y-(1-\epsilon)|| y<1-\epsilon\right]=1.
			\end{aligned}
		\end{equation*}
		
		Hence, 
        
		\begin{equation*}
			\begin{aligned}
				\mathbb{E}_{y\sim f(\textbf{c}')}[SS(y,\textbf{c}')]&=\mathbb{E}_{y\sim f(\textbf{c}')}[SS(y,\textbf{c}')|y<1-\epsilon](1-q)\\
                &+\mathbb{E}_{y\sim f(\textbf{c}')}[SS(y,\textbf{c}')| y\geq 1-\epsilon]q\\
				&\leq (1-q)+SS(opt(\textbf{c}'),\textbf{c}')q\\
				&=(1-q)+\frac{1}{1-\epsilon}q=1+\frac{\epsilon}{1-\epsilon}q.
			\end{aligned}
		\end{equation*}
        
		From $q\leq \frac{1}{2(1-2\epsilon)}$, we can get the approximation ratio
		\begin{equation*}
			\alpha \geq \frac{opt(\textbf{c}')}{\mathbb{E}_{y\sim f(\textbf{c}')}[SS(y,\textbf{c}')]}=\frac{\frac{1}{1-\epsilon}}{1+\frac{\epsilon}{1-\epsilon}q}\geq \frac{\frac{1}{1-\epsilon}}{1+\frac{\epsilon}{1-\epsilon}\cdot\frac{1}{2(1-2\epsilon)}}=\frac{2(1-2\epsilon)}{2(1-\epsilon)(1-2\epsilon)+\epsilon}.
		\end{equation*}
		
		When $\epsilon=\frac{1}{2}-\frac{\sqrt{2}}{4}$, $\alpha\geq \frac{4\sqrt{2}}{4+\sqrt{2}}\approx1.0448$ which implies any randomized strategy-proof  mechanism has an approximation ratio of at least $1.0448$.\qed
\end{proof}

\subsection{Discussions}

For maximizing the minimum  satisfaction, Mechanism \ref{th3} has unbound approximation ratio for the sum-variant satisfaction. Consider a location profile $\textbf{x}$ with two agents, where $\textbf{x}_1=(0,0)$ and $\textbf{x}_2=(1,1)$. Mechanism \ref{th3} locates the facility at 0, which gives the minimum satisfaction of 0. However, it holds that $opt(\textbf{x})=\frac{1}{2}$ and $MS(opt(\textbf{x}),\textbf{x})=\frac{1}{2}$. Then an unbounded approximation ratio is derived. 

In fact, any deterministic mechanism whose output corresponds to a location that minimizes the distance for some agent has an unbounded approximation ratio. We conjecture that any deterministic strategy-proof mechanism gives  unbounded approximation ratio for maximizing the minimum  of the sum-variant satisfaction and the minimum of  the max-variant satisfaction. 

\section{Conclusions and Future Work}
In this paper, we studied  mechanism design for desirable and obnoxious facility location games with satisfaction and multi-location agents. For the individual objective, we studied two types by different distance functions:  the sum-variant satisfaction and the max-variant satisfaction. We considered two social objectives: maximizing  the social satisfaction, and maximizing the minimum satisfaction. We proposed some deterministic and randomized (group) strategy-proof mechanisms with small approximation ratios for our model. 

There are many other interesting problems for future researches. In terms of our results, it would be interesting to  close the gap between upper and lower bounds of  deterministic and randomized strategy-proof mechanisms. One can also generalize our model in $k$-facility location problem, where $k \geq 2$. Additionally, it might make sense to extend the line setting to trees, circles or general cases for  multi-location agents. Finally, other preference models (for example,  fractional preference \cite{fong2018facility}) can be discussed  with satisfaction. 
\subsubsection{\ackname} This research was supported in part by the National Natural Science Foundation of China (12201590, 12171444) and Natural Science Foundation of Shandong Province (ZR2024MA031). We thank the anonymous reviewers  for their helpful discussions.

\subsubsection{\discintname} The authors declare that they have no conflict of interest. 

\bibliographystyle{splncs04}
\bibliography{main}
\newpage
\noindent
{\large\bfseries A Appendix}
\vspace{0.5cm}

\noindent
\textbf{Proof of Proposition 1.}
\begin{proof}
    Given a location profile $\textbf{x}$ and a facility location $y$,  it suffices to  prove that the ratio 
 $ \frac{d_2(y,\textbf{x}_i)-\delta(\textbf{x}_i)}{\Delta(\textbf{x}_i)-\delta(\textbf{x}_i)} $ depends solely on the midpoint $c_i$  according to the definition of satisfaction of agent $i$. 

For agent $i$, we have that 
\begin{equation*}
\begin{aligned}
     \frac{d_2(y,\textbf{x}_i)-\delta(\textbf{x}_i)}{\Delta(\textbf{x}_i)-\delta(\textbf{x}_i)}
     &=\frac{(|y-c_i|+\frac{x_{i\omega_i}-x_{i1}}{2})-\frac{x_{i\omega_i}-x_{i1}}{2}}{\max_{z\in I}(|z-c_i|+\frac{x_{i\omega_i}-x_{i1}}{2})-\frac{x_{i\omega_i}-x_{i1}}{2}}\\
    &=\frac{|y-c_i|}{\max_{z\in I}|z-c_i|}=\frac{|y-c_i|}{\max\{c_i,1-c_i\}}.
\end{aligned}
\end{equation*}

The proof is completed.\qed
\end{proof}

\noindent
\textbf{Proof of Theorem 1.}
\begin{proof}
    Firstly, let us prove that Mechanism \ref{me1} is group strategy-proof. Let $f$ denote Mechanism \ref{me1}. For  any location profile $\textbf{x}$, let $G\subseteq N$ be any coalition of agents. $\textbf{x}'_G$ is a joint deviation of the agents in $G$. If $f(\textbf{x})=f(\textbf{x}'_G,\textbf{x}_{-G})$, no agent can improve her satisfaction. If $f(\textbf{x})\neq f(\textbf{x}'_G,\textbf{x}_{-G})$, we assume $f(\textbf{x})<f(\textbf{x}'_G,\textbf{x}_{-G})$, there exists an agent $i\in G$ who meets $med(\textbf{x}'_i)<f(\textbf{x})\leq med(\textbf{x}_i)$. $med(\textbf{x}_i)$ is the location where the distance sum  between agent's locations and  facility location is the minimum. Then, we have that $$\sum_{j=1}^{\omega_i}{|x_{ij}-med(\textbf{x}'_i)}|>\sum_{j=1}^{\omega_i}{|x_{ij}-f(\textbf{x})|}\geq \sum_{j=1}^{\omega_i}{|x_{ij}-med(\textbf{x}_i)}|.$$ Then, $s(f(\textbf{x}'_G,\textbf{x}_{-G}),\textbf{x}_i)<s(f(\textbf{x}),\textbf{x}_i)$. Therefore, Mechanism \ref{me1} is group strategy-proof.  
    
    Next we show the approximation ratio of Mechanism 1.   For $i=1,\dots,\lfloor\frac{n}{2}\rfloor$, let $\textbf{p}_i=(\textbf{x}_i,\textbf{x}_{n-i+1})$ and $SS(y,\textbf{p}_i)$ be the location profile and the sum of the satisfactions of agents $i$ and $n-i+1$, respectively. 
	
	If $n$ is even, then $SS(y, \textbf{x}) = SS(y, \textbf{p}_1) +\dots + SS(y, \textbf{p}_{\frac{n}{2}}),$ which implies that $$\frac{SS(opt(\textbf{x}), \textbf{x})}{SS(f(\textbf{x}), \textbf{x})}\leq \max_{i=1,\dots,\frac{n}{2}}\left\{\frac{SS(opt(\textbf{x}), \textbf{p}_i)}{SS(f(\textbf{x}), \textbf{p}_i)}\right\}.$$
	
	If $n$ is odd, $SS(y, \textbf{x}) = SS(y, \textbf{p}_1) + \cdots +SS(y, \textbf{p}_{\lfloor\frac{n}{2} \rfloor}) + s(y, x_{\lceil\frac{n}{2}\rceil})$. Analogously, we get that $$\frac{SS(opt(\textbf{x}), \textbf{x})}{SS(f(\textbf{x}), \textbf{x})}\leq \max\left\{ \frac{SS(opt(\textbf{x}), \textbf{p}_i)}{SS(f(\textbf{x}), \textbf{p}_i)} |_{i=1,\dots,\lfloor\frac{n}{2}\rfloor},\frac{s(opt(\textbf{x}), \textbf{x}_{\lceil\frac{n}{2}\rceil})}{s(f(\textbf{x}), \textbf{x}_{\lceil\frac{n}{2}\rceil})}\right\}.$$
	
	 By $f(\textbf{x})=med(\textbf{x}_{\lceil\frac{n}{2}\rceil})$ and $s(med(\textbf{x}_{\lceil\frac{n}{2}\rceil}), \textbf{x}_{\lceil\frac{n}{2}\rceil})=1$,   we have  that $$\frac{s(opt(\textbf{x}), \textbf{x}_{\lceil\frac{n}{2}\rceil})}{s(f(\textbf{x}), \textbf{x}_{\lceil\frac{n}{2}\rceil})} \leq 1.$$
	
	For $i=1,\dots,\lfloor\frac{n}{2}\rfloor$, let $opt(\textbf{p}_i)$ denote the optimal facility location when we only consider  agents $i$ and $n-i+1$.  Note that $SS(opt(\textbf{x}), \textbf{p}_i) \leq SS(opt(\textbf{p}_i),\textbf{p}_i) \leq2$ and $f(\textbf{x})\in  [med(\textbf{x}_i),med( \textbf{x}_{n-i+1})]$. Hence, we get that for $i=1,\dots,\lfloor\frac{n}{2}\rfloor$, 
    \begin{equation}\label{eq11}
        \frac{SS(opt(\textbf{x}), \textbf{p}_i)}{SS(f(\textbf{x}), \textbf{p}_i)}\leq \frac{SS(opt(\textbf{p}_i), \textbf{p}_i)}{\min_{y\in[med(\textbf{x}_i),med( \textbf{x}_{n-i+1})]}SS(y, \textbf{p}_i)}.
    \end{equation}
	
	It is sufficient to show that for  $i=1,\dots,\lfloor\frac{n}{2}\rfloor$, the right side of Eq. (\ref{eq11}) is at most 2.
	
	Let $\hat{y}_i$ be the location where  $SS(y, \textbf{p}_i)$ achieves its minimum on the interval $[med(\textbf{x}_i),med( \textbf{x}_{n-i+1})]$, i.e., $\hat{y}_i= \arg \min_{y\in [med(\textbf{x}_i),med( \textbf{x}_{n-i+1})]}{SS(y, \textbf{p}_i)}$. We have that $$SS(y,\textbf{p}_i)=\frac{\Delta(\textbf{x}_i)-\sum_{j=1}^{\omega_i} |y-x_{ij}|}{\Delta(\textbf{x}_i)-\delta(\textbf{x}_i)}+\frac{\Delta(\textbf{x}_{n-i+1})-\sum_{j=1}^{\omega_{n-i+1}} |y-x_{ij}|}{\Delta(\textbf{x}_{n-i+1})-\delta(\textbf{x}_{n-i+1})}.$$

   Since  $SS(med(\textbf{x}_i),\textbf{p}_i)\geq1$, $SS(med(\textbf{x}_{n-i+1}),\textbf{p}_i)\geq1$ and $SS(y, \textbf{p}_i)$ is a concave function  of $y$, we have $SS(\hat{y}_i, \textbf{p}_i)\geq 1$. Therefore, $$\frac{SS(opt(\textbf{x}), \textbf{p}_i)}{SS(f(\textbf{x}), \textbf{p}_i)} \leq 2.$$
    
	Moreover, the approximation ratio of Mechanism \ref{me1} is tight. Consider a location profile $\textbf{x}$ with two agents, where  $\textbf{x}_1=(0,\frac{1}{2})$ and $\textbf{x}_2=(\frac{1}{2},1)$. For this profile, we have that $opt(\textbf{x})=\frac{1}{2}$ and $SS(opt(\textbf{x}),\textbf{x})=2$. However, Mechanism $1$ locates the facility at $0$ and achieves $SS(f(\textbf{x}),\textbf{x})=1$. Thus, the approximation ratio of Mechanism 1 is $2$. \qed
\end{proof}

\noindent
\textbf{Proof of Theorem \ref{th2}.}

\begin{proof}
     Mechanism 2 is obviously group strategy-proof. Hence we only need to show that the approximation ratio $\alpha$ of Mechanism 2 is 2.

     Let $f$ denote Mechanism \ref{me2}. For any  \textbf{x}, we known that  
     $MS(f(\textbf{x}),\textbf{x})=\min_{i\in N}s(f(\textbf{x}),\textbf{x}_i)$ and $MS(opt(\textbf{x}),\textbf{x})\leq 1$.  If we can show that for each agent $i$, $$s(f(\textbf{x}),\textbf{x}_i)=\frac{\max_{z\in I}\sum_{j=1}^{\omega_i}|z-x_{ij}|-\sum_{j=1}^{\omega_i}|\frac{1}{2}-x_{ij}|}{\max_{z\in I}\sum_{j=1}^{\omega_i}|z-x_{ij}|-\sum_{j=1}^{\omega_i}|med(\textbf{x}_i)-x_{ij}|}\geq\frac{1}{2},$$ then the proof is completed. 
    
    Without loss of generality, we assume that $\sum_{j=1}^{\omega_i}x_{ij}\geq\sum_{j=1}^{\omega_i}(1-x_{ij})$. Then, 
    \begin{equation*}
        \begin{aligned}   \sum_{j=1}^{\omega_i}x_{ij}&\geq\sum_{j=1}^{\omega_i}|x_{ij}-(1-med(\textbf{x}_i)|\geq2 \sum_{j=1}^{\omega_i}\left|x_{ij}-\frac{1}{2}\right|-\sum_{j=1}^{\omega_i}|x_{ij}-med(\textbf{x}_i)|.\\
        \end{aligned}
    \end{equation*}

    Moreover, we have that  $$ \sum_{j=1}^{\omega_i}x_{ij}-\sum_{j=1}^{\omega_i}\left|x_{ij}-\frac{1}{2}\right|\geq\frac{1}{2}\left(\sum_{j=1}^{\omega_i}x_{ij}-\sum_{j=1}^{\omega_i}|x_{ij}-med(\textbf{x}_i)|\right).$$

      Similarly, if $\sum_{j=1}^{\omega_i}x_{ij}<\sum_{j=1}^{\omega_i}(1-x_{ij})$, then  $$ \sum_{j=1}^{\omega_i}(1-x_{ij})-\sum_{j=1}^{\omega_i}\left|x_{ij}-\frac{1}{2}\right|\geq\frac{1}{2}\left(\sum_{j=1}^{\omega_i}(1-x_{ij})-\sum_{j=1}^{\omega_i}|x_{ij}-med(\textbf{x}_i)|\right).$$
    
We have that $s(f(\textbf{x}),\textbf{x}_i)\geq\frac{1}{2}$ for any agent $i\in N$. Then, it holds that $MS(f(\textbf{x}),\textbf{x})\geq\frac{1}{2}$ and $\frac{MS(opt(\textbf{x}),\textbf{x})}{MS(f(\textbf{x}),\textbf{x})}\leq 2$.

Moreover, the approximation ratio of Mechanism 2 is tight. Consider a location profile $\textbf{x}$ with two agents, where  $\textbf{x}_1=(0,\frac{1}{2})$ and $\textbf{x}_2=(0,0)$. For this profile $\textbf{x}$, we have that $opt(\textbf{x})=0$ and $MS(opt(\textbf{x}),\textbf{x})=1$. However,  $MS(\frac{1}{2},\textbf{x})=\frac{1}{2}$. Therefore, the  approximation ratio of Mechanism 2 is $2$. \qed
\end{proof}

\noindent
\textbf{Proof of Theorem 4.}
\begin{proof}
    Similar to the proof of Theorem 2, Mechanism 2 achieves that the max-variant satisfaction of each agent is at least $\frac{1}{2}$. For any location profile $\textbf{x}$, it holds that $MS(opt(\textbf{x}),\textbf{x})\leq1$. Therefore, the approximation ratio of Mechanism 2 is no more than 2. 
    
    Furthermore, this ratio can be verified to be tight. Consider a instance $\textbf{x}$ with $n$ agents, where each agent has a location at $0$. Then, we have $opt(\textbf{x})=0$, yielding $MS(opt(\textbf{x}),\textbf{x})=1$. The facility location output by Mechanism 2 is $\frac{1}{2}$, giving the minimum satisfaction of $\frac{1}{2}$. Hence, the approximation ratio of Mechanism 2 is 2. \qed
\end{proof}

\noindent
\textbf{Proof of Theorem 5.}

\begin{proof}
    Firstly, we show that Mechanism \ref{th3} is a group strategy-proof mechanism. Let $f$ denote Mechanism \ref{th3}. Given a location profile $\textbf{x}$, each agent in $N_1(\textbf{x})$ gets satisfaction of $1$ if $y=0$, each agent in $N_2(\textbf{x})$ gets satisfaction of $1$ if $y=1$. Let $f(\textbf{x})$ be the facility location returned by Mechanism \ref{th3} and assume $f(\textbf{x})=0$. If there exists a coalition $G$ which changes the output of mechanism by misreporting their locations, namely  $f(\textbf{x}'_G, \textbf{x}_{-G})=1$, there must exists an agent in $N_1(\textbf{x})$ but she reports falsely in $N_2(\textbf{x})$ and then  her satisfaction strictly decreases.
    
    The approximation ratio of mechanism \ref{th3} is at most 2 since at least half of  agents will achieve their best possible satisfaction, that is, $SS(f(\textbf{x}),\textbf{x})\geq\frac{n}{2}$. With $SS(opt(\textbf{x}),\textbf{x})\leq n$, we get that $\alpha\geq2$.

    Moreover, the approximation ratio of Mechanism \ref{th3} is tight. Consider a location profile $\textbf{x}$ with two agents, where $\textbf{x}_1=(0,1)$ and $\textbf{x}_2=(0,\frac{1}{2})$. For this profile, we can show that $opt(\textbf{x})=1$ and  $SS(opt(\textbf{x}),\textbf{x})=2$. However, Mechanism \ref{th3} locates the facility at $0$ and achieves  $SS(f(\textbf{x}),\textbf{x})=1$. Thus, the approximation ratio of Mechanism 4 is $2$.  \qed
\end{proof}

\noindent
\textbf{Proof of Theorem 7.}
    \begin{proof}
    Let $f$ be a randomized strategy-proof mechanism which has approximation ratio of $\alpha$. Consider a location profile $\textbf{x}$ with two agents, where $\textbf{x}_1=(\frac{1}{6},\frac{1}{6},\frac{5}{6})$, $\textbf{x}_2=(\frac{5}{6},\frac{5}{6},\frac{5}{6})$. The optimal facility location is at 0 and the social satisfaction is $SS(opt(\textbf{x}),\textbf{x})=\frac{10}{7}.$

The social satisfaction  obtained from $f$ is
\begin{equation}\label{eq1}
    \begin{aligned}
        SS(f(\textbf{x}),\textbf{x})&=\mathbb{E}_{y\sim f(\textbf{x})}[s(y,\textbf{x}_1)+s(y,\textbf{x}_2)]\\
        &=\mathbb{E}_{y\sim f(\textbf{x})}\left[\frac{2|y-\frac{1}{6}|+|y-\frac{5}{6}|-\frac{2}{3}}{\frac{11}{6}-\frac{2}{3}}+\frac{3|y-\frac{5}{6}|}{\frac{5}{2}}\right]\\
        &=\mathbb{E}_{y\sim f(\textbf{x})}\left[\frac{12}{7}\left|y-\frac{1}{6}
        \right|+\frac{72}{35}\left|y-\frac{5}{6}\right|\right]-\frac{4}{7}\geq\frac{1}{\alpha}\cdot\frac{10}{7}.
    \end{aligned}
\end{equation}
In Eq. (\ref{eq1}), 
	\begin{equation} \label{eq3}
    \begin{aligned}
        \mathbb{E}_{y\sim f(\textbf{x})}\left[\left|y-\frac{1}{6}\right|\right]&=\mathbb{E}_{y\sim f(\textbf{x})}\left[\frac{1}{6}-y\bigg| y\leq\frac{1}{6}\right]\\
        &+\mathbb{E}_{y\sim f(\textbf{x})}\left[y-\frac{1}{6}\bigg|\frac{1}{6}<y<\frac{5}{6}\right]+\mathbb{E}_{y\sim f(\textbf{x})}\left[y-\frac{1}{6}\bigg|y\geq\frac{5}{6}\right],
    \end{aligned}
		\end{equation}

\begin{equation}\label{eq4}
    \begin{aligned}
       \mathbb{E}_{y\sim f(\textbf{x})}\left[\left|y-\frac{5}{6}\right|\right]
        &=\mathbb{E}_{y\sim f(\textbf{x})}\left[\frac{5}{6}-y\bigg|y\leq\frac{1}{6}\right]\\
        &+\mathbb{E}_{y\sim f(\textbf{x})}\left[\frac{5}{6}-y\bigg|\frac{1}{6}<y<\frac{5}{6}\right]+\mathbb{E}_{y\sim f(\textbf{x})}\left[y-\frac{5}{6}\bigg|y\geq\frac{5}{6}\right].
    \end{aligned}
\end{equation}

Let $L_1=\mathbb{E}_{y\sim f(\textbf{x})}[\frac{1}{6}-y|y\leq \frac{1}{6}]$, $R_1=\mathbb{E}_{y\sim f(\textbf{x})}[y-\frac{5}{6}|y\geq\frac{5}{6}]$. From Eqs. (\ref{eq3}) and (\ref{eq4}), we can get 
\begin{equation}\label{eq5}
    \begin{aligned}
        \mathbb{E}_{y\sim f(\textbf{x})}\left[\left|y-\frac{1}{6}\right|+\left|y-\frac{5}{6}\right|\right]=2L_1+2R_1+\frac{2}{3}.
    \end{aligned}
\end{equation}
Since for any $y\in I$, $\frac{2}{3}\leq|y-\frac{1}{6}|+|y-\frac{5}{6}|\leq1$, we have  $L_1,R_1\geq0, L_1+R_1\leq \frac{1}{6}$.
Then, 
\begin{equation}\label{eq22}
    SS(f(\textbf{x}),\textbf{x})=\frac{144}{35}(L_1+R_1)+\frac{4}{5}-\frac{12}{35}\mathbb{E}_{y\sim f(\textbf{x})}\left[y-\frac{1}{6}\right]\geq \frac{1}{\alpha}\cdot\frac{10}{7}.
\end{equation}

From Eq. (\ref{eq22}) and $L_1+R_1\leq \frac{1}{6}$, we can get \begin{equation}\label{eq6}
    \mathbb{E}_{y\sim f(\textbf{x})}\left[y-\frac{1}{6}\right]\leq\frac{13}{3}-\frac{25}{6}\cdot\frac{1}{\alpha}.
\end{equation}

Consider a new profile $\textbf{x}'$ with two agents where $\textbf{x}'_1=(\frac{1}{6},\frac{1}{6},\frac{1}{6})$ and $\textbf{x}'_2=(\frac{5}{6},\frac{5}{6},\frac{5}{6})$, which implies agent 1 in $\textbf{x}$ misreport her locations to $(\frac{1}{6},\frac{1}{6},\frac{1}{6})$. The optimal facility location of $\textbf{x}'$ is at 0 and its social satisfaction is $\frac{6}{5}$. The social satisfaction of $\textbf{x}'$ followed by $f$ is 
\begin{equation}\label{eq2}
    \begin{aligned}
        SS(f(\textbf{x}'),\textbf{x}')&=\mathbb{E}_{y\sim f(\textbf{x}')}[s(y,\textbf{x}_1')+s(y,\textbf{x}_2')]\\
        &=\mathbb{E}_{y\sim f(\textbf{x}')}\left[\frac{3|y-\frac{1}{6}|}{\frac{5}{2}}+\frac{3|y-\frac{5}{6}|}{\frac{5}{2}}\right]\\
        &=\frac{6}{5}\mathbb{E}_{y\sim f(\textbf{x}')}\left[\left|y-\frac{1}{6}\right|+\left|y-\frac{5}{6}\right|\right]\geq\frac{1}{\alpha}\cdot\frac{6}{5}.
    \end{aligned}
\end{equation}

Let $L_2=\mathbb{E}_{y\sim f(\textbf{x}')}[\frac{1}{6}-y|y\leq \frac{1}{6}]$, $R_2=\mathbb{E}_{y\sim f(\textbf{x}')}[y-\frac{5}{6}|y\geq\frac{5}{6}]$. Similar to Eq. (\ref{eq5}), we can get 
\begin{equation}\label{eq7}
     \mathbb{E}_{y\sim f(\textbf{x}')}\left[\left|y-\frac{1}{6}\right|+\left|y-\frac{5}{6}\right|\right]=2L_2+2R_2+\frac{2}{3}.
\end{equation}

Then, $$SS(f(\textbf{x}'),\textbf{x}')=\frac{12}{5}(L_2+R_2)+\frac{2}{5}\geq\frac{6}{5}\cdot\frac{1}{\alpha}$$ which implies  $L_2+R_2\geq\frac{1}{2\alpha}-\frac{1}{3}$. 

Since $f$ is a randomized strategy-proof mechanism, we have that  $s(f(\textbf{x}),\textbf{x}_1)\geq s(f(\textbf{x}'),\textbf{x}_1)$, which implies $\mathbb{E}_{y\sim f(\textbf{x})}[2|y-\frac{1}{6}|+|y-\frac{5}{6}|]\geq\mathbb{E}_{y\sim f(\textbf{x}')}[2|y-\frac{1}{6}|+|y-\frac{5}{6}|]$.
	Due to Eqs. (\ref{eq5}) and (\ref{eq7}), we have 
	\begin{equation}\label{eq8}
		\mathbb{E}_{y\sim f(\textbf{x})}\left[\left|y-\frac{1}{6}\right|\right]+2L_1+2R_1\geq \mathbb{E}_{y\sim f(\textbf{x}')}\left[\left|y-\frac{1}{6}\right|\right]+2L_2+2R_2.
	\end{equation}
From Eqs. (\ref{eq6}) and (\ref{eq8}) and $L_1+R_1\leq \frac{1}{6}$, we have 
$\mathbb{E}_{y\sim f(\textbf{x}')}\left[|y-\frac{1}{6}|\right]+2L_2+2R_2\leq \frac{14}{3}-\frac{25}{6}\cdot\frac{1}{\alpha}$.  By symmetric, we have $\mathbb{E}_{y\sim f(\textbf{x}')}\left[|y-\frac{5}{6}|\right]+2L_2+2R_2\leq \frac{14}{3}-\frac{25}{6}\cdot\frac{1}{\alpha}$.

	With the last two inequalities and  $L_2+R_2\geq\frac{1}{2\alpha}-\frac{1}{3}$,  we can get the following, $$\mathbb{E}_{y\sim f(\textbf{x}')}\left[\left|y-\frac{1}{6}\right|+\left|y-\frac{5}{6}\right|\right]\leq\frac{28}{3}-\frac{25}{3}\cdot\frac{1}{\alpha}-4\left(\frac{1}{2\alpha}-\frac{1}{3}\right)\leq \frac{32}{3}-\frac{31}{3}\cdot\frac{1}{\alpha}.$$
	The approximation ratio of $f$ is 
	\begin{equation} \label{eq9}
		\alpha\geq\frac{SS(opt(\textbf{x}'),\textbf{x}')}{SS(f(\textbf{x}'),\textbf{x}')}\geq \frac{\frac{6}{5}}{\frac{6}{5}\left[d(p_2,\frac{1}{6})+d(p_2,\frac{5}{6})\right]}\geq\frac{1}{\frac{32}{3}-\frac{31}{3}\cdot\frac{1}{\alpha}}.   
	\end{equation}
	
	Eq. (\ref{eq9}) shows that $\alpha\geq\frac{34}{32}\approx1.0625$. \qed

\end{proof}

\end{document}